\documentclass[11pt]{article}
\usepackage[top=50mm, bottom=50mm, left=50mm, right=50mm]{geometry}
\usepackage{tocbibind}
\usepackage{amssymb}
\usepackage{amsmath}
\usepackage{amsthm}
\usepackage{epsfig}
\usepackage{graphicx}
\usepackage{graphics}
\usepackage{float}
\usepackage{subfigure}
\usepackage{multirow}
\usepackage{color}
\usepackage{lineno}
\usepackage{fullpage}
\usepackage[normalem]{ulem} 
\usepackage{makeidx}
\usepackage{xspace}
\usepackage{wrapfig}
\makeindex

\newtheorem{theorem}{Theorem}

\newtheorem{corollary}{Corollary}

\newtheorem{definition}{Definition}

\newtheorem{note}{Note}

\newtheorem{lemma}{Lemma}

\definecolor{darkred}{rgb}{1, 0.1, 0.3}
\definecolor{darkblue}{rgb}{0.1, 0.1, 1}
\definecolor{darkgreen}{rgb}{0,0.6,0.5}

\newcommand {\mm}[1] {\ifmmode{#1}\else{\mbox{\(#1\)}}\fi}

\newcommand{\myTo}	{{{T}_{1}^{f}}}
\newcommand{\myTt}	{{{T}_{2}^{g}}}

\newcommand{\uTo}		{{|T_1^f|}}
\newcommand{\uTt}		{{|T_2^g|}}

\begin{document}
\title{Frechet-Like Distances between Two Merge Trees}
\author{Elena Farahbakhsh Touli }

\maketitle
\setcounter{page}{0}

\begin{abstract}
The purpose of this paper is to extend the definition of Frechet distance which measures the distance between two curves to a distance (Frechet-Like distance) which measures the similarity between two rooted trees. The definition of Frechet-Like distance is as follows: Tow men start from the roots of two trees. When they reach to a node with the degree of more than $2$, they construct $k-1$ men which $k$ is the outgoing degree of the node and each man monitor a man in another tree (there is a rope between them). The distance is the minimum length of the ropes between the men and the men whom are monitored and they all go forward (the geodesic distance between them to the root of the tree increases) and reach to the leaves of the trees.
Here, I prove that the Frechet-Like distance between two trees is SNP-hard to compute.

I modify the definition of Frechet-Like distance to measure the distance between tow merge trees, and I prove the relation between the interleaving distance and the modified Frechet-Like distance.
\end{abstract}

\newpage
\section{Introduction}
In this paper I are interested in extending the definition of the Frechet distance between curves to a distance between two trees. 

Frechet distance between curves is a distance for measuring the similarity between two curves. For the first time Frechet distance was defined by Maurice Fréchet \cite{CtFdbtpc, f15, f17}. Later, Frechet distance attracted attention and was worked on by other people \cite{CtdFdis, TkFd, CtFdbtpc, f1, CtFdwsiN}. 

The intuitive definition of Frechet distance between two curves is as follows: A man and his dog start from the starting points of two curves and  a leash connects the dog to the man. They can only go forward. The Frechet distance between the curves is the minimum length of the leash that the man and the dog start from the beginning of the curves and they reach to the end of the curves without separating the leash. In the following I write the mathematical definition of the Frechet distance between two curves \cite{CtFdbtpc}.

\begin{definition}\cite{CtFdbtpc}
Suppose that I have two curves $C_1: [a,b]\longrightarrow V$ and $C_2:  [a',b']\longrightarrow V$,  such that $a<b$ and $a'< b'$ and $V$ is a vector space. The Frechet distance between $C_1$ and $C_2$ is defined as the infimum distance over all continuous increasing functions $\alpha: [0,1]\longrightarrow[a,a']$ and $\beta: [0,1] \longrightarrow[b,b']$ that maximizes the distance between $C_1(\alpha(t))$ and $C_2(\beta(t))$ on $t\in [0,1]$. In this case, the Frechet distance is defined as follows: 
$$d_F (C_1, C_2)= \inf_{\alpha, \beta}\max_{t\in[0,1]}\{d(C_1(\alpha(t)), C_2(\beta(t)))\}.$$
\end{definition}

Weak Frechet distance is a special kinds of Frechet distance such that the man and the dog can go backward as well \cite{CtFdbtpc}. Both Frechet distance and weak Frechet distance can be found in a polynomial time between two polygonal curves \cite{CtFdbtpc}, but it is NP-hard to compute the Frechet distance between two surfaces \cite{Ontcomtbgoihd} and till now no one has defined Frechet distance between trees. Discrete Frechet distance was discussed by T. Eiter and H. Mannila in 1994 \cite{CdFd}. In 2012, P.K.Agrawal, etc. found an algorithm to find the discrete Frechet distance between two polygonal curves in sub-quadratic time. \cite{CtdFdis}

 \noindent{\bf New work.} In this paper I will extend the definition of the Frechet distance between curves to define a similar distance between rooted trees.

This is the first time that the Frechet distance is defined between trees. I call it Frechet-Like distance because of the similarity of this definition to the Frechet distance between curves. The intuitive definition of the Frechet-Like distance is as follows:
Two men $A$ and $B$ start form the roots of merge trees $T_1$ and $T_2$ respectively and there is a rope between them. When they reach to a node (a vertex with the degree of greater than $2$) each of them construct $(k-1)$ men  similar to themselves, which $k$ is the outgoing degree of the node. Each man from $T_1$ ($T_2$) is assigned to walk simultaneously (there is a rope between them) with a man in $T_2$ ($T_1$); he can stop in a node or go forward with a man from $T_2$ ($T_1$) or the man in $T_2$ ($T_1$) stops and another one goes forward. A man from $T_1$ ($T_2$) can walk with more than one man from $T_2$ ($T_1$) only if the distance between them is not more than $\varepsilon$ and the man in $T_1$ ($T_2$) stops in one point. Also if two men $A$ and $B$ go simultaneously with two men $A'$ and $B'$, respectively, there is a rope between the nearest common ancestor or $A$ and $B$ and also the nearest common ancestor of $A'$ and $B'$.  The problem is finding the minimum length of the rope that the man can start from the beginning of the tree $T_1$ and another from $T_2$ and there is at least a man to reach to any leaf of $T_1$ and a man that reaches to any leaf of $T_2$. A man should either construct other men or reach to a leaf. In this definition length of the rope with one end point in $x\in |T_1|$ and one end point in $y\in |T_2|$ is defined as $d(x,y)$. Here I consider the Euclidean distance as the distance between $x$ and $y$. 

Later I modify the definition of Frechet-Like distance to a definition between two merge trees.
By considering the merge trees $\myTo$ and $\myTt$, I prove the relation between the modified Frechet-Like distance  and  the interleaving distance between two merge trees. 
\begin{definition}
Merge tree. \cite{Idbmt,Fafcghaidbt}

A merge tree is a rooted tree with a function which is defined on each point of the tree. A merge tree $T^h$ is defined by a pair $(T,h)$ such that $h: |T|\longrightarrow \mathbb{R}$ is a monotone function which means that if for $x,y\in |T|$ $x< y\footnote{$x < y$ is that $x$ is a descendant of $y$}$, $h(x)<h(y)$. 

Intuitively I can define a merge tree $(T,h)$  as follows: consider a tree and a node of the tree as the node $u$. Hang the tree from the node. I consider the function value $h(u) = 0$ for $u$ that I hang the tree from and for all the other points in the tree, the function of each point of the merge tree $T^h_u$ will be the negative distance between the node $u$ and the point.
\end{definition}

The outcome of this paper is as follows: The distance between two trees is discussed in section \ref{DbT}. In section \ref{FLDbMT}, I define the Frechet-Like distance between trees, both the intuition and mathematical definition of Frechet-Like distance. In section \ref{AoFLDiiN}, I prove that it is NP-hard to approximate the Frechet-Like distance between rooted trees. Section \ref{d} is considered for modifying the Frechet-like distance between two merge trees. I also  prove the relation between the interleaving distance and the modified Frechet-like distance between two merge trees in this section. Section \ref{C} is the conclusion.


\section{Distance between Trees}\label{DbT}
Distance between trees is one of the topics that has been discussed in the previous years \cite{f12,f14,Aotaatte,Idbmt, Fafcghaidbt}. The tree edit distance and the tree alignment distance are two well-known distances which were defined between trees \cite{Aotaatte}. Both the tree edit distance and the tree alignment distance between two trees are MAX SNP-hard to compute. There is a polynomial time algorithm for computing the tree alignment distance between two ordered trees\footnote{Ordered tree is a rooted tree that there is an order between the children of each node \cite{Aotaatte}.} if we bound the degree of each node, however there is no known polynomial algorithm for finding the edit distance between ordered trees with bounded degrees. There is a polynomial algorithm for computing the tree edit distance between trees if we consider trees with bounded depth \cite{Aotaatte}.

\begin{definition}
Tree edit distance \cite{Aotaatte}.

Consider two labeled trees $T_1$ and $T_2$. The tree edit distance is the minimum cost of changing one tree to another one by using three editing operations add, remove and rename. 
\end{definition}

\begin{definition}
Tree alignment distance \cite{Aotaatte}.

Consider two labeled trees $T_1$ and $T_2$. The alignment distance between the two trees is obtained as follows: 
first I add nodes to $T_1$ and $T_2$ that the modified trees $T'_1$ and $T'_2$ have the same structures. The related cost would be the the cost of changing the labels that two trees $T'_1$ and $T'_2$ have also same labels. The minimum cost related to the best structural changes is the alignment distance.
\end{definition}

Two following notes are satisfied about the tree edit distance and tree alignment distance from \cite{Aotaatte} and \cite{Fafcghaidbt} respectively.

\begin{note}
Tree alignment distance is always greater than or equal to tree edit distance. For more illumination, look at Figure \ref{fig:trees} (b).
\end{note}

\begin{note}
Although there is a polynomial time algorithm for finding tree alignment distance between ordered labeled trees, tree alignment distance cannot capture similarities between trees. Figure \ref{fig:trees} (a) illustrates this better. 
\end{note}

\begin{figure}[tbph] 
\begin{center}
\begin{tabular}{ccccc}
\includegraphics[height=4cm]{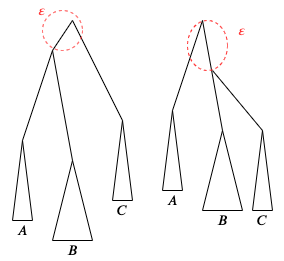} &\hspace*{0.15in} &
 \includegraphics[height=3.5cm]{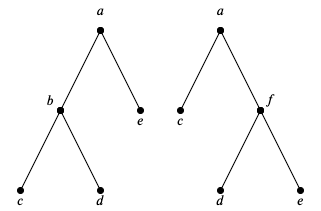} & \hspace*{0.15in}
& 
\includegraphics[height=4cm]{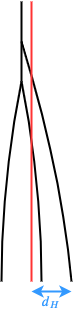} \\
(a) & & (b) & & (c) 
\end{tabular}
\end{center}
\vspace*{-0.5cm}
\caption{
(a) Two trees are very similar to each other, but the alignment distance between them is very large, because tree alignment distance is sensible to the relationship between children and parents.
(b) If the costs of relabeling, removing and adding nodes is $1$, tree edit distance between two trees is 2, and tree alignment distance between them is 4.
(c) Two trees (red color and black color trees) are completely different, however the Hausdorff distance between them is small. 
\label{fig:trees}}
\end{figure}

Another distance that we can consider between trees is Hausdorff distance \cite{f1}. Hausdorff distance is defined between two sets of points. The Hausdorff distance is defined as follows: 

\begin{definition}  Hausdorff Distance \cite{f1}. 

For given sets $S_1$ and $S_2$ in a space, for each point $s$ in $S_1$ we find the closest point to it in $S_2$ (as $s'$), and for each point in $S_2$ we find the closest point to it in the set $S_1$. The Hausdorff distance is the maximum over all distances that we find. The mathematical definition of Hausdorff distance is as follows:

$$
d_H(S_1,S_2)=\max\left\{ \sup_{s\in S_1} \inf_{s'\in S_2} \text{d}(s,s'),\sup_{s'\in S_2} \inf_{s\in S_1}\text{d}(s,s')\right\}.
$$
\end{definition}
If we consider the underlying space of trees on $\mathbb{R}^2$, we can define Hausdorff distance between two trees. 
However, the Hausdorff distance cannot capture dissimilarities between trees. For example in Figure \ref{fig:trees} the two trees are very different, however the Hausdorff distance between them is very small.

Another distance that we can consider between trees is interleaving distance\cite{Idbmt, Fafcghaidbt}. Interleaving distance  is defined between merge trees. Interleaving distance between two merge trees $\myTo$ and $\myTt$ is defined by two continuous functions $\alpha$ and $\beta$ and the definition is as follows: 

\begin{definition}\label{dI1}\cite{f12, Idbmt, Fafcghaidbt}
Interleaving distance between two merge trees $\myTo$ and $\myTt$ is defined as follows: 
$$
d_{I}(\myTo, \myTt) = inf\{\delta \text{ s.t. there is a pair of $\delta$-compatible maps between $\myTo$ and $\myTt$}\},
$$
where two continuous maps $\alpha_\delta: |\myTo|\longrightarrow |\myTt|$ and $\beta_\delta: |\myTt|\longrightarrow |\myTo|$ are $\delta$-compatible if and only if the following conditions are satisfied:
\newline
(1) For all $u\in|\myTo|$, $g(\alpha_\delta(u))=f(u) + \delta$, 
\newline
(2) For all $v\in|\myTt|$, $f(\alpha_\delta(v))=g(v) + \delta$,
\newline
(3) For all $u_1, u_2\in|\myTo|$ s.t. $f(u_1)=f(u_2)$, $\beta_{\delta} o \alpha_\delta(u_1) = \beta_{\delta} o \alpha_\delta(u_2)) = u_1^{2\delta}$,
\newline
(4) For all $v_1, v_2\in|\myTt|$ s.t. $g(v_1)=g(v_2)$, $\alpha_{\delta} o \beta_{\delta}(v_1)= \alpha o\beta_\delta(v_2) = v_1^{2\delta}$.
\end{definition}

In \cite{f12}, P. K. Agrawal, etc., proved that it is NP-hard to compute interleaving distance between two merge trees and it concludes the fact that it it NP-hard to compute the Gromov-Hausdorff distance between trees within a factor of better than $3$. Later in 2019 E. Farahbakhsh and Y. Wang \cite{Fafcghaidbt} defined one $\varepsilon$-good map from $\myTo$ to $\myTt$ which is defined as follows: 

\begin{definition}\label{ID}\cite{Fafcghaidbt}
A map $\alpha^{\delta}: |\myTo| \longrightarrow |\myTt|$ is called $\delta$-good map if and only if the following conditions are satisfied: 
\newline
(C1) $\alpha^{\delta}$ is continuous,
\newline
(C2) For every point $u\in|\myTo|$, $g(\alpha^{\delta}(u))= f(u) + \delta$,
\newline
(C3) For every pair of points $v_1 = \alpha^{\delta}(u_1)$ and $v_2 = \alpha^{\delta}(u_2)$, if $v_1\geq v_2$, $u_1^{2\delta}\footnote{$u_1^{2\delta}$ is an ancestor of $u_1$ in $\myTo$ such that $f(u_1^{2\delta}) - f(u_1)= 2\delta$}\geq u_2^{2\delta}$,
\newline
(C4) If there is a point $v\in |\myTt|$ which is not in the image of $\alpha^{\delta}$, $f(v^F\footnote{$v^F$ is the nearest ancestor of $v$ such that $v^F$ is in the image of $\alpha^{\delta}$.}) - f(v)\leq 2\delta$.
\end{definition}
and by the definition of $\delta$-good map, they proved the following Theorem: 
\begin{theorem}\cite{Fafcghaidbt}
$d_I(\myTo, \myTt)\leq \delta$ if and only if there is a $\delta$-good map $\alpha^{\delta}:|\myTo|\longrightarrow |\myTt| $.
\end{theorem}

\section{Frechet-Like Distance between two Rooted Trees}\label{FLDbMT}
In this section I define the Frechet-Like distance between two rooted trees. Given two merge trees $T_1$ and $T_2$ rooted at $u$ and $v$ respectively, the definition of Frechet-Like distance is that I consider two men $A$ and $B$ who start to walk from the points $u$ and $v$ respectively. Two men are connected by using a rope. If a man reaches to a node (with a degree of higher than 2) he will copy himself $k-1$ times such that $k$ is the outgoing degree of the node. Each man at each time can monitor (there is a rope between them) just  one man at a time unless the man stops and others go far away just $\varepsilon$ distance. Here the distance between two men is defined the distance between their function value of the merge tree at those points that they are. The Frechet-Like distance is defined as the minimum distance between the point that two men are. In the following I write the mathematical definition of the Frechet-Like distance:

\begin{definition}\textbf{Frechet-Like Distance}\label{def1}

For two given rooted trees $T_1$ and $T_2$, I define Frechet-Like distance as follows: 
\begin{equation*}
    d_{FL}(\myTo,\myTt):= \min_{R\in \mathcal{R}} \sup_{(x,y)\in R}d(x,y)
\end{equation*}
$d(x,y)$ is the Euclidean distance between two points $x$ and $y$ and the correspondence $R \subseteq \uTo \times \uTt$ is defined as follows:
\\
\\
1) $\forall x\in |T_1|$, $\exists y\in |T_2|$ s.t. $(x, y)\in R$

1-i) $\forall y\in |T_2|$, $\exists x\in |T_1|$ s.t. $(x, y)\in R$\\
2) If $(x_1, y_1)\in R$ and $(x_2, y_2) \in R$ and $x_2 \geq x_1$  and $y_2\geq y_1$ then

2-i) $\forall x$ s.t. $x_1\leq x\leq x_2$, $\exists y$ s.t. $y_1\leq y\leq y_2$ and $(x, y)\in R$ and 

2-ii) $\forall y$ s.t. $y_1\leq y\leq y_2$, $\exists x$ s.t. $x_1\leq x\leq x_2$ and $(x, y)\in R$.\\
3) If $(x_1, y_1)\in R$ and $(x_2, y_2) \in R$ then $(x_1 \sim x_2\footnote{$x_1 \sim x_2$ is the nearest ancestor of $x_1$ and $x_2$},  y_1\sim y_2) \in R$.\\
4) If $x\in |T_1|$ is a leaf, there should be a leaf $y\in|T_2|$ such that $(x,y)\in R$, unless there is a $y'$ such that $(x,y')\in R$ and $(x^N\footnote{$x^N$ the nearest node which is an ancestor of $x$},y')\in R$.

4-i) If $y\in |T_2|$ is a leaf, there should be a leaf $x\in|T_1|$ such that $(x,y)\in R$, unless there is a $x'$ such that $(x',y)\in R$ and $(x',y^N)\in R$.
\end{definition}

\section{Approximation of the Frechet-Like Distance is in NP-hard}\label{AoFLDiiN}
In this section I prove that computing the Frechet like distance between two rooted trees is SNP-hard to compute by a reduction from UNRESTRICTED-PARTITION. The way that I prove that it is in SNP-hard is very similar to proving that Gromov-Hausdorff distance between two merge trees is in SNP-complete. \cite{f12}
\\
\\
\textbf{UNRESTRICTED-PARTITION.} 
\\
Input: a multiset of positive integers $\mathcal{X} = \{a_1,...,a_n\}$ such that $n = 3k$,
\\
Output: Is there a partition of $\mathcal{X}$ into $k$ multisets ${\mathcal{X}_1, . . . , \mathcal{X}_m}$ such that for each multiset $\mathcal{X}_j$ if we consider by $S_j$ the summation of elements in multiset $\mathcal{X}_j$, 
$S_j = (\sum_{i=1}^n a_i) /m$? \cite{f13}

\begin{theorem}
The problem UNRESTRICTED-PARTITION is in SNP-complete.
\end{theorem}
\begin{proof}
See (\cite{f13}).
\end{proof}

Here, I construct two merge trees $\myTo$ and $\myTt$ as follows. In the following picture $A$ and $B$ are two large numbers. 

\begin{figure}[htbp]
\begin{center}
\includegraphics[height=5cm]{D1.png}
\end{center}
\vspace*{-0.15in}
\caption{Two trees $\myTo$ and $\myTt$. $A$ and $B$ are two large numbers.
\label{fig:D1}}
\end{figure}

if I consider two merge trees $\myTo$ and $\myTt$ that their roots is located in one point, $d(x,y) = |f(x)-g(y)|$ in the Definition\ref{def1}. Now, I prove the hardness of approximation of Frechet-Like distance by the following lemmas. 

\begin{lemma}
$d_{FL}(f, g) \leq 1$ if UNRESTRICTED-PARTITION is a yes instance. 
\end{lemma}
\begin{proof}
If UNRESTRICTED-PARTITION is a yes instance, I can construct a correspondence $R\subseteq \uTo\times\uTt$ such that 
$\sup_{(x,y)\in R} |f(x)-g(y)|\leq1$. If UNRESTRICTED-PARTITION is a yes instance, I can partition $\mathcal{X}$ into $X_1, X_2,...,X_n$ such that $S(X_1)=\frac{S(\mathcal{X})}{k}$, and $X_i=\{a_{i,1},...,a_{i,k_i}\}$. Therefore, I map sub-trees rooted at $\{u_{i_1},...,u_{i_{k_i}}\}$ to $v_i$, such that $u_{i_j}$ corresponds to $a_{i,j}$ in the construction of the tree and $v_i$ corresponds to $X_i$. When I say that I map a point $x\in \uTo$ to a point $y\in \uTt$, we mean that $(x, y) \in R$
If $(u_{i_j},v_i)\in R$ and $(u_{i_k},v_i)\in R$, I have that $(u_r,v_i)\in R$. (For more illustration look at Figure \ref{fig:D1}) Therefore, I could construct a correspondence $R\subseteq \uTo\times\uTt$ such that 
$\sup_{(x,y)\in R} |f(x)-g(y)|<1$.
\begin{figure}[htbp]
\begin{center}
\includegraphics[height=6cm]{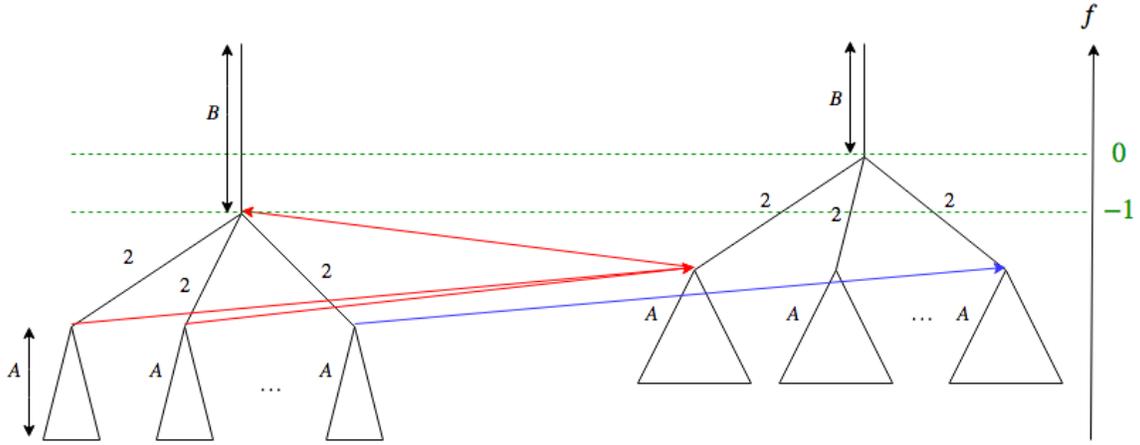}
\end{center}
\vspace*{-0.15in}
\caption{if UNRESTRICTED-PARTITION is a yes instance.
\label{fig:D2}}
\end{figure}
\end{proof}

\begin{lemma}
 If UNRESTRICTED-PARTITION is a no instance, $d_{FL}(f, g) \geq 3$.
\end{lemma}
\begin{proof}
If UNRESTRICTED-PARTITION is a no instance, as edges with the length of $A$ are too large, we have to find a correspondence $R$ such that for any pair of points $x_1, x_2\in\myTo$ such that $x_1\|x_2\footnote{$x_1\|x_2$ if $x_1 \nleq x_2$ nor $x_2\nleq x_1$}$, there are two different points $y_1, y_2\in\myTt$ such that $(x_1, y_1)\in R $, and $(x_2, y_2)\in R$. Therefore, the best correspondence that I can find with the conditions of the definition \ref{def1} is that $x_1 \sim x_2$ map to two different point $y_1$ and $y_2$ as shown in Figure \ref{fig:D3}. Which indicates that the Frechet distance between $\myTo$ and $\myTt$ cannot be smaller than 3.
\end{proof}
\begin{figure}[htbp]
\begin{center}
\includegraphics[height=6cm]{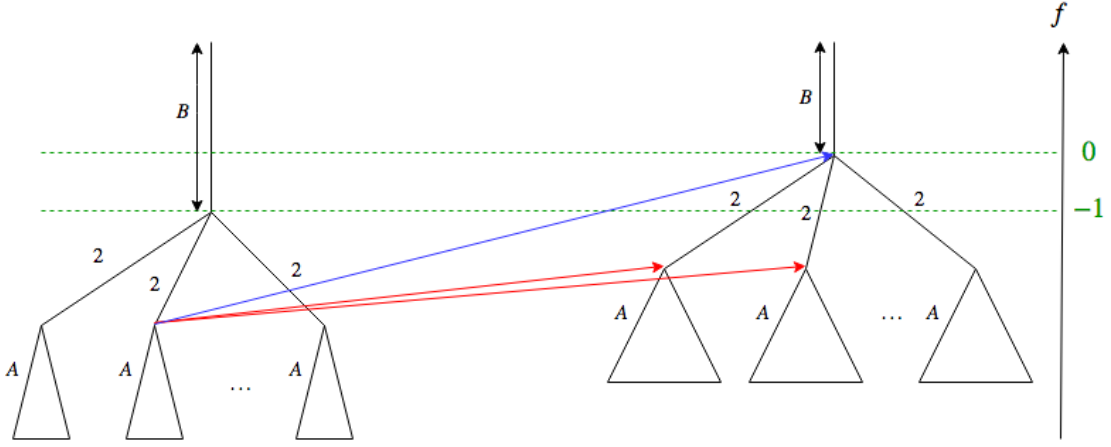}
\end{center}
\vspace*{-0.15in}
\caption{if UNRESTRICTED-PARTITION is a yes instance.
\label{fig:D3}}
\end{figure}

From two mentioned lemmas, we can conclude the following result.  
\begin{corollary} 
Computing a $(3-\epsilon)$-approximation of the Frechet-Like distance between two merge trees $\myTo$ and $\myTt$ is NP-complete,
\end{corollary}

\section{Frechet-Like distance and the interleaving distance}\label{d}
In this section we define a Frechet-Like distance between two merge trees, which we call Frechet-Like distance between merge trees. Given two merge trees $\myTo$ and $\myTt$ rooted at $u$ and $v$ respectively, the definition of  Frechet-Like distance  is defined as the minimum distance between the function value of the point that a man is and the function value of the point that the monitored man is. In the following we write the mathematical definition of modified Frechet-Like distance between two merge trees:

\begin{definition}\textbf{Frechet-Like Distance}\label{def1}

For two given merge trees $\myTo$ and $\myTt$, we define Frechet-Like distance as follows: 
\begin{equation*}
    d_{FL}^M(\myTo,\myTt):= \min_{R\in \mathcal{R}} \sup_{(x,y)\in R} |f(x)-g(y)|
\end{equation*}
and the correspondence $R \subseteq \uTo \times \uTt$ is defined like the Definition \ref{dI1}.

\end{definition}

By the following lemma, we prove the relation between the Frechet-like distance and the interleaving distance between merge trees.
\begin{lemma}
If there exists an $\varepsilon$ such that $d_{FL}^M(\myTo, \myTt)\leq\varepsilon$, then
$d_I(\myTo, \myTt)\leq\varepsilon$.
\end{lemma}

\begin{proof}
For proving this lemma we need to find an $\varepsilon$-good map $\alpha_\varepsilon: \uTo\longrightarrow \uTt$ such that three conditions in the definition of $\varepsilon$-good map are satisfied. First, we consider the $\varepsilon$-good map $\alpha^\varepsilon$ as follows: 

As the Frechet-Like distance between $\myTo$ and $\myTt$ is not greater than $\varepsilon$, based on the definition \ref{def1} there is a correspondence $R$ such that four conditions of the definition \ref{def1} are satisfied. Now for constructing of the $\varepsilon$-good map for any pair of points $(x,y)\in R$ if $g(y) = f(x) + \varepsilon$, we map the point $x$ to $y$, in another words $\alpha^\varepsilon(x) = y$. Otherwise, if $g(y) < f(x) + \varepsilon$, we map $x$ to a point $y'$ such that $y\leq y'$ and $g(y') = f(x) + \varepsilon$, it means that $\alpha_\varepsilon(x) = y'$. 

Now, we need to prove that $\alpha_\varepsilon$ is an $\varepsilon$-good map. To do so, we need to prove that four conditions of the definition \ref{ID} for the map $\alpha_\varepsilon$ are satisfied. 
\\
{\bf{C1.}} We need to prove that map $\alpha_\varepsilon$ is continuous. To do so, we use the similar method as is written in \cite{Fafcghaidbt}. 
\\
{\bf{C2.}} Based on the construction of the map $\alpha^\varepsilon$ for any pair of points $(x,y)\in R$ we map $x$ to a point which is $\varepsilon$ distance higher than $x$. As for all the point $x$ in $|\myTo|$ there is at least one $y$ such that $(x,y)\in R$, we can conclude that for all the point $x$ in $|\myTo|$, $g(\alpha^{\varepsilon}(x))=f(x)+ \varepsilon$, which satisfies the condition (C2) of the Definition \ref{ID}.
\\
{\bf{C3.}} If two pairs of points $(x_1, y_1)\in R$ and $(x_2, y_2)\in R$, and $y_1\leq y_2$, we know that $g(y_1)\leq g(y_2)$. Therefore, based on the construction of the map $\alpha^\varepsilon$, we have that $f(x_1)\leq f(x_2)$. Two cases can happen: 

Case1: $x_1\leq x_2$, which in this case we have that $x_1^{2\epsilon}\leq x_2^{2\epsilon}$. 

Case2: $x_1\| x_2$, in this case if by contradiction $x_1^{2\epsilon}\nleq x_2^{2\epsilon}$, therefore we have that $x_1^{2\epsilon}\| x_2^{2\epsilon}$ as $f(x_1)\leq f(x_2)$. Based on the definition of Frechet-Like distance the highest $y$ such that $(x_2,y)\in R$ is $y_2$. Therefore by using the condition 3 of the Frechet-Like distance the highest point $y in \myTt$ that $(x_1\sim x_2,y)\in R$ is $y_2$ and $f(x_1\sim x_2) - g(y_2)>\varepsilon$. It is a contradiction with the fact that the Frechet-Like distance between $\myTo$ and $\myTt$ is less than or equal to $\varepsilon$.
\begin{figure}[htbp]
\begin{center}
\includegraphics[height=8cm]{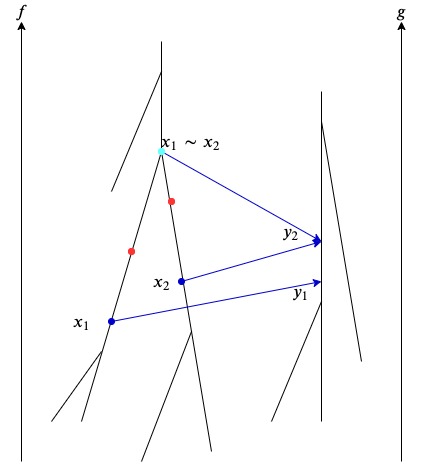}
\end{center}
\vspace*{-0.15in}
\caption{
\label{fig:case2}}
\end{figure}
\\
{\bf{C4.}} If there is a point $y\in \myTo$ such that there is no $x\in \myTt$ map to $y$ under the map $\alpha_\varepsilon$, as we already proved in C1 that the map is continuous, the point should be a branch connects a leaf (For example $y^L$) to the tree, and none of  the point $y'\leq y$ are in the image of the map $\alpha_\varepsilon$. 
Now, I just need to prove that $g(y^F)-g(y)\leq 2\varepsilon$. By contradiction if $g(y^F)-g(y)> 2\varepsilon$ and $x$ is the point that $(x, y^F)\in R$ based on the definition of Frechet-Like distance condition 4, $x$ is a leaf. Therefore, $(x,y^L)\in R$ and it is a contradiction by the fact that  $d_{FL}^M(\myTo, \myTt)\leq\varepsilon$.
\end{proof}

\section{Concluding Remarks}\label{C}
In this paper I extended the Frechet distance between two curves to Frechet-Like distance between rooted trees. In section \ref{DbT}, I discussed some distances that have been defined between two trees. I defined a new definition for computing the similarity between two trees in Section \ref{DbT}. I called the new distance, Frechet-Like distance because of the similarity of the definition to Frechet distance between curves. The hardness of approximation was discussed later in Section \ref{AoFLDiiN}. Here, we also proved that although there is a polynomial time algorithm for computing the Frechet distance between polygonal curves \cite{CtFdbtpc}, it is NP-hard to approximate Frechet-Like distance between two trees. The relation between Frechet-Like distance between two merge trees and the interleaving distance was discussed in section \ref{d}.
\bibliography{bibil}

\end{document}